\documentclass[a4paper,11pt]{article}
\usepackage{amsmath,amssymb,amsthm}
\usepackage[pdftex]{graphicx}
\usepackage[pdfpagemode=UseOutlines,
 bookmarks=true,bookmarksopen=true,bookmarksnumbered=true,
 pdffitwindow=true,pdfpagelayout=SinglePage,pdfhighlight=/P,
 colorlinks=true,linkcolor=blue,citecolor=blue,urlcolor=blue]{hyperref}
\usepackage[indention=0mm,justification=centerlast,
 font=small,labelfont=bf,labelsep=period]{caption}

\theoremstyle{plain}
\newtheorem{theorem}{Theorem}
\newtheorem{statement}[theorem]{Statement}
\theoremstyle{definition}
\newtheorem{definition}{Definition}
\theoremstyle{remark}
\newtheorem{remark}{Remark}

\newcommand{\sect}[1]{\subsubsection*{\em\mdseries #1}}
\newcommand{\nn}{\mathbf n}

\title{On discrete integrable equations of higher order}

\author{V.E.~Adler\thanks{L.D.~Landau Institute for Theoretical Physics, 1A
Ak.~Semenov, Chernogolovka 142432, Russia. E-mail: adler@itp.ac.ru} ,\quad
V.V.~Postnikov\thanks{Sochi Branch of Peoples' Friendship University of Russia,
32 Kuibyshev str., 354000 Sochi, Russia. E-mail: postnikofvv@mail.ru}}

\date{3 October 2013}

\begin{document}\maketitle

\begin{abstract}\noindent
We study 2D discrete integrable equations of order $1$ with respect to one
independent variable and $m$ with respect to another one. A generalization of
the multidimensional consistency property is proposed for this type of
equations. The examples are related to the B\"acklund--Darboux transformations
for the lattice equations of Bogoyavlensky type.
\medskip

\medskip

\end{abstract}

\section{Introduction}\label{s:intro}

The paper is devoted to study of integrable equations with two discrete
independent variables, of the form
\begin{equation}\label{full_Qm}
 Q(v(i+1,n+m),\dots,v(i+1,n),v(i,n+m),\dots,v(i,n))=0
\end{equation}
where $m$ is a fixed positive integer, dependent variable $v$ and function $Q$
take real or complex values. We define the property of multidimensional
consistency for such equations and illustrate it by two examples related with
the Darboux transformations for spectral problems of order $m+1$. Recall that,
for the so-called {\em quad-equations}
\begin{equation}\label{Q1}
 Q(v(i+1,n+1),v(i+1,n),v(i,n+1),v(i,n))=0,
\end{equation}
3D-consistency means that a generic set of 2-dimensional initial data on a
3-dimensional lattice defines the function $v(i,j,n)$ which satisfies
simultaneously three equations of the form (\ref{Q1}), with respect to each
pair of discrete variables \cite{BS,NW,ABS03}. This automatically implies the
consistency on the lattice of arbitrary dimension.

A generalization of this notion for {\em multiquad-equations} (\ref{full_Qm})
is given in Section \ref{s:cons}. In this case, the variable $n$ is
distinguished and the situation is less symmetric: equations of the form
(\ref{full_Qm}) are fulfilled with respect to the variables $(i,n)$, $(j,n)$
and the variables $(i,j)$ correspond to some $m$-component quad-equation
\begin{equation}\label{R}
 V(i+1,j+1)=R(V(i+1,j),V(i,j+1),V(i,j))
\end{equation}
where $V=(v(n+m),\dots,v(n+1))$. Symbolically,
\[
 \text{quad}+\text{quad}+\text{quad}\quad\overset{1\,\to\,m}{\longrightarrow}\quad
 m\text{-quad}+m\text{-quad}+m\text{-component quad}.
\]
In the multidimensional lattice, equation (\ref{R}) satisfies the usual
3D-consistency property with respect to variables $i,j,k$ others than $n$.

\sect{Examples}

Equations (\ref{full_Qm}) appear in the theory of B\"acklund transformations
(or Darboux transformations, on the level of spectral problems) for
evolutionary differential-difference equations
\begin{equation}\label{ut}
 \partial_tu=A(u_m,\dots,u_{-m})
\end{equation}
where the notation $u_s=u(n+s,t)$ is used. Undoubtedly, most well-known
equation of this type is the Bogoyavlensky lattice \cite{N,I,B}
\[
 \partial_tu=u(u_m+\dots+u_1-u_{-1}-\dots-u_{-m}).
\]
Its B\"acklund transformation was obtained in papers \cite{THO,PN,S96}; a
detailed account can be found in book \cite{S03} (an approach based on the
discretization preserving Hamiltonian structure); among recent publications, we
mention \cite{FYIIN} (a relation with the generalized QD algorithm). In Section
\ref{s:B}, we reproduce this B\"acklund transformation in the form of an
equation of type (\ref{full_Qm}) which is consistent with a potential version
of the Bogoyavlensky lattice. In this context, the variable $n$ in
(\ref{full_Qm}) is inherited from the lattice equation and $i$ enumerates the
B\"acklund transformations. A new result in Section \ref{s:B} is the derivation
of consistent equation (\ref{R}) which corresponds to the nonlinear
superposition of the B\"acklund transformations. This equation turns out to be
a $m$-component reduction of 3D equation of Hirota type.

It should be noted that lattice equations of the form (\ref{ut}) are well
studied only at $m=1$. In this case, an exhaustive classification of equations
admitting higher symmetries was obtained by Yamilov \cite{Y83,Y06}. Relation of
such lattice equations with quad-equations is also well studied
\cite{LPSY,LY09,LY11,GGH,GY,MWX}. At $m>1$, search of new examples and study of
their properties remain an actual open problem. The theory of lattice equations
(\ref{ut}) is more complicated than the parallel theory of continuous
integrable evolutionary equations; this is explained by the fact that a single
continuous equation may correspond to an infinite family of discretizations of
different order. For instance, the Korteweg--de Vries equation is obtained
under the continuous limit from the Bogoyavlensky lattices for arbitrary $m$,
and the lattices corresponding to the different $m$ are not related to each
other and belong to different hierarchies.

In Section \ref{s:dSK}, we derive a new example of equation of the from
(\ref{full_Qm}). It defines the Darboux--B\"acklund transformation for the
nonhomogeneous generalization of the Bogoyavlensky lattice from our previous
article \cite{AP}. This lattice equation serves as a discretization of the
Sawada--Kotera equation and admits the Lax representation with the operator $L$
equal to the ratio of two difference operators. It is interesting that the
Darboux transformation posesses a similar structure. As in the case of the
Bogoyavlensky lattice, the nonlinear superposition principle brings to the
consistency with certain $m$-component equation of the form (\ref{R}).

\sect{Notations and assumptions}

A point in the multidimensional integer lattice is denoted as
$\nn=(n_1,n_2,\dots,n)$. The last coordinate $n$ is distinguished, the shifts
$T^s: n\to n+s$ with respect to this variable are denoted by the subscript $s$
and the zero subscript is omitted, like in equation (\ref{ut}). For the other
coordinates, we consider only the unit shifts $T_i:n_i\to n_i+1$ which are
denoted by the superscript $i$. In this notation, equation (\ref{full_Qm}) with
variable $i$ replaced by $n_i$ is written as
\begin{equation}\label{Qm}
 Q(v^i_m,\dots,v^i,v_m,\dots,v)=0
\end{equation}
and this is the form which we will use in what follows. We assume that this
equation is solvable with respect to any of corner variables $v^i_m,v^i,v_m$ or
$v$, for the generic values of the rest variables involved in the
equation\footnote{Depending on the context, the term `variable' is used either
for a function defined on the lattice, or for its value in one point.}. So, we
consider equation (\ref{Qm}) equivalent to equation of the form
\begin{equation}\label{q}
 v^i_m=q(v^i_{m-1},\dots,v^i,v_m,\dots,v),
\end{equation}
and analogously for the variables $v^i,v_m,v$. Moreover, we assume that the
nondegeneracy condition is fulfilled
\[
 \partial_{v^i}q\not\equiv0,\quad
 \partial_{v_m}q\not\equiv0,\quad
 \partial_vq\not\equiv0,
\]
in order to exclude from the consideration equations of the form $Q=Q'Q''=0$
where each factor depends on incomplete subset of corner variables. In fact,
the examples presented in this paper correspond to the case when $Q$ is an
irreducible polynomial of power 1 with respect to each its argument. It is easy
to see that in such a case all above stipulations are fulfilled.

In addition to the variable $v$ defined in the vertices of the integer lattice,
we consider also the variables defined on its edges. A variable associated with
the edge $(\nn,T_i(\nn))$ is denoted like $f^{(i)}$ (and it should be
distinguished from $f^i=T_i(f)$).

\sect{The general scheme}

In both examples, the dressing procedure is quite standard (it is similar to
the Veselov--Shabat approach in the continuous case \cite{VS}). Given a
discrete spectral problem of order $m$
\[
 L[u]\psi=\lambda\psi,
\]
the Darboux transformation is constructed by use of its particular solution
$\phi$ at $\lambda=\alpha$. It brings, for the function $f=\phi_1/\phi$, to the
pair of Miura type substitutions
\begin{equation}\label{uuf}
 u=a(f_m,\dots,f,\alpha),\quad \tilde u=b(f_m,\dots,f,\alpha)
\end{equation}
and the sequence of these substitutions (a discrete dressing chain) is
described by the equation
\[
 a(\tilde f_m,\dots,\tilde f,\tilde\alpha)=b(f_m,\dots,f,\alpha)
\]
where tilde is understood as the shift with respect to the second discrete
variable. Although this equation belongs to the type (\ref{Qm}), it {\em is
not} 3D-consistent in the above sense, because the variables $f$ are associated
with the edges of the lattice rather than the vertices, and in this situation
one should use another definition which generalizes the notion of Yang--Baxter
mappings, see e.g. \cite{V,ABS04,PTV,PT}. Certainly, both versions of
3D-consistency, for the vertices and the edges, are closely related. It turns
out that equations (\ref{uuf}) admit introducing of the potential (due to some
conservation law)
\[
 u=u(v_m,v),\quad f=f(\tilde v,v)
\]
and the new variable $v$ satisfies an equation of the form (\ref{Qm}) which
falls under our definition of consistency. The additional $m$-component
quad-equation (\ref{R}) corresponding to the nonlinear superposition of Darboux
transformations is found from the matrix representation.

It should be remarked that B\"acklund transformations equivalent to a pair of
Miura type substitutions are very common, but do not cover all known examples.
For instance, this scheme does not contain the quad-equation $Q_4$ which is the
most general known equation at $m=1$. This equation defines the B\"acklund
transformation for the elliptic Volterra lattice, and also the superposition of
B\"acklund transformation for the Krichever--Novikov equation. One may expect
that analogous examples exist for $m>1$ as well, but these are not discovered
yet.

\section{Multidimensional consistency}\label{s:cons}

Let us consider the 3-dimensional integer lattice $(n_i,n_j,n)$ and let the
real-valued variable $v$ defined on the lattice satisfies the
$m$-quad-equations with respect to each pair of discrete variables $n_i,n$ and
$n_j,n$:
\begin{equation}\label{QQ}
 Q^{(i)}(v^i_m,\dots,v^i,v_m,\dots,v)=0,\quad
 Q^{(j)}(v^j_m,\dots,v^j,v_m,\dots,v)=0.
\end{equation}
The solution of these equations on the coordinate sublattices $(n_i,0,n)$ and
$(0,n_j,n)$, with the generic initial data
\[
 v(0,0,n),\quad
 v(n_i,0,0),\dots,v(n_i,0,m-1),\quad
 v(0,n_j,0),\dots,v(0,n_j,m-1),
\]
can be constructed by solving the equations with respect to the corner
variables (see fig.~\ref{fig:ini} which illustrates the case $m=2$). The
extension of these solutions on the whole 3D lattice requires certain
compatibility conditions which we will analyze now.

\begin{figure}[b]
\centerline{\includegraphics[width=0.6\textwidth]{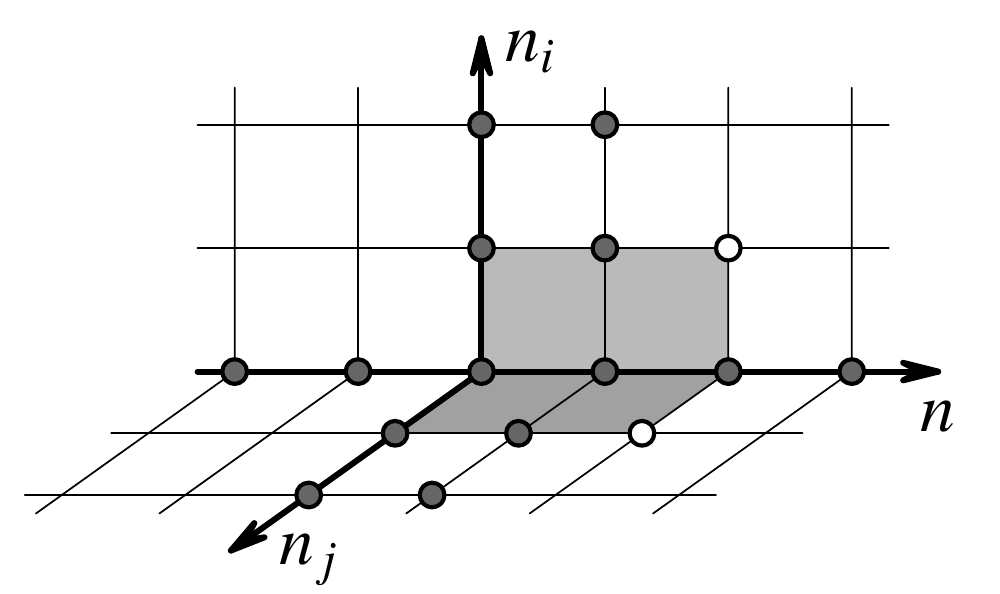}}
\captionsetup{width=0.6\textwidth}
\caption{Pair of equations (\ref{QQ}), $m=2$. Initial data are marked with
the dark discs.}
\label{fig:ini}
\end{figure}

For a moment, it is convenient to identify the reference point $\nn$ on the
lattice with the origin $(0,0,0)$. The values $v^i_s=v(1,0,s)$,
$v^j_s=v(0,1,s)$ are computed from the initial data according to equations
(\ref{QQ}). In order to find the values $v^{ij}_s=v(1,1,s)$, we have to solve
the infinite set of equations
\begin{equation}\label{n11}
\begin{aligned}[b]
 &Q^{(i)}(v^{ij}_{s+m},\dots,v^{ij}_s,v^j_{s+m},\dots,v^j_s)=0,\\
 &Q^{(j)}(v^{ij}_{s+m},\dots,v^{ij}_s,v^i_{s+m},\dots,v^i_s)=0,\quad
  s\in\mathbb Z.
\end{aligned}
\end{equation}
Taking $s=m-1,\dots,0$, we get a system of $2m$ equations for $2m$ unknowns
$v^{ij}_{2m-1},\dots,v^{ij}$. We will assume that functions $Q^{(i)},Q^{(j)}$
are generic in the sense that this system is not degenerate and possesses a
finite set of solutions. However, this does not guarantee that a solution
exists for the whole set of equations (\ref{n11}). Indeed, if we consider
additionally $s=m$ then two new equations for a single new unknown
$v^{ij}_{2m}$ are added, which, generally, do not admit a solution in common.
We are interested in the special type of equations such that a solution of
system (\ref{n11}) exists for the generic initial data.

Notice that expressions for $v^{ij}_{m-1},\dots,v^{ij}$ as functions of initial
data can contain only
\[
 v_{m-1},\dots,v,\quad v^i_{m-1},\dots,v^i,\quad v^j_{m-1},\dots,v^j.
\]
Indeed, since $v^{ij}_{m-1},\dots,v^{ij}$ are found by solving equations
(\ref{n11}) at $s=m-1,\dots,0$, hence these variables do not depend on $v_s$ at
$s<0$. Similarly, these variables can be found by solving equations (\ref{n11})
at $s=-1,\dots,-m$, and this implies that there are no dependence on $v_s$ at
$s>m-1$, as well. Therefore, if equations (\ref{QQ}) are consistent then a
mapping
\[
 R^{(ij)}:~(V^i,V^j,V)\to V^{ij},\quad V=(v_{m-1},\dots,v)
\]
appears, which can be interpreted as a $m$-component quad-equation on the
sublattice $(n_i,n_j)$. We derived it in the origin of the lattice, but it is
clear that this mapping is defined at any point $\nn$.

\begin{figure}[t]
\centerline{\includegraphics[width=\textwidth]{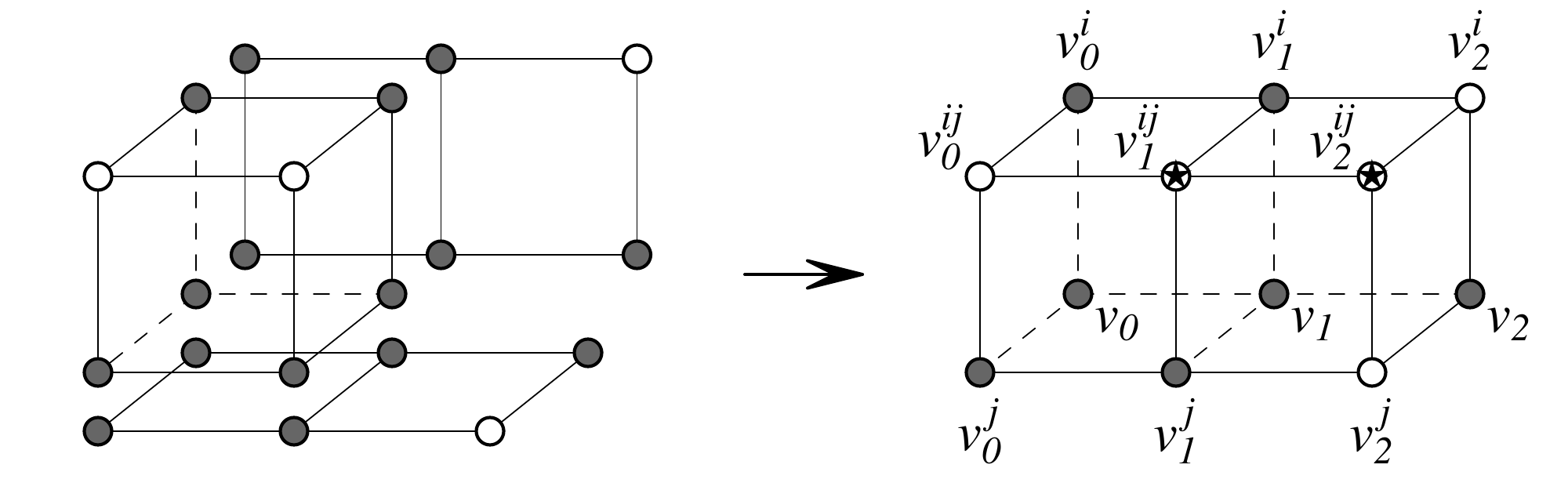}}
\captionsetup{width=0.75\textwidth}
\caption{3D-consistency for a pair of 2-quad-equations and a 2-component
quad-equation. Dark disks correspond to the initial data; the stars mark the
points where the consistency conditions are tested.}
\label{fig:3Dcons}
\end{figure}

This observation brings to the desired formulation of the consistency
conditions. In the following Definition, equations (\ref{QQ}) are taken in the
resolved form (\ref{q}) (so that the solution of the 3D-consistent system is
constructed in the octant $n_i\ge0$, $n_j\ge0$, $n\ge0$). The components of the
mapping $R^{(ij)}$ are enumerated by a subscript put in the square brackets in
order to distinguish it from the shift with respect to $n$. The consistency
condition on the 3-dimensional lattice is given in part (i) of the definition
(fig.~\ref{fig:3Dcons} illustrates it for $m=2$). In the multidimensional case,
this condition must be satisfied on the sublattices $(n_i,n_j,n)$ for all
$i,j$. The mappings $R^{(ij)}$ defined on all sublattices $(n_i,n_j)$ must
satisfy the usual 3D-consistency condition for quad-equations on the
sublattices $(n_i,n_j,n_k)$, as stated in the part (ii).

\begin{definition}\label{def:cons}
Let $V=(v_{m-1},\dots,v)$, $R^{(ij)}=(r^{(ij)}_{[m-1]},\dots,r^{(ij)}_{[0]})$.
The system of equations
\begin{gather}
\label{vim}
 v^i_m=q^{(i)}(v^i_{m-1},\dots,v^i,v_m,\dots,v)=q^{(i)}(V^i,v_m,V),\\
\label{vij}
 V^{ij}=R^{(ij)}(V^i,V^j,V),\quad i\ne j
\end{gather}
is called multidimensionally consistent if the following properties are
fulfilled:

(i) for any pair $i\ne j$, the relations
\begin{equation}\label{consij}
\begin{gathered}
 r^{(ij)}_{[s]}(V^i_1,V^j_1,V_1)
  =r^{(ij)}_{[s+1]}(V^i,V^j,V),\quad s=0,\dots,m-2,\\
 r^{(ij)}_{[m-1]}(V^i_1,V^j_1,V_1)
  =q^{(i)}(V^{ij},v^j_m,V^j)
  =q^{(j)}(V^{ij},v^i_m,V^i)
\end{gathered}
\end{equation}
where $v^i_m$, $v^j_m$, $V^{ij}$ are substituted from (\ref{vim}), (\ref{vij}),
hold identically with respect to $v_m,V$, $V^i$, $V^j$;

(ii) for any triple $i\ne j\ne k\ne i$, the relations
\begin{equation}\label{consijk}
  R^{(jk)}(V^{ij},V^{ik},V^i) =R^{(ik)}(V^{ij},V^{jk},V^j)
 =R^{(ij)}(V^{ik},V^{jk},V^k)
\end{equation}
where $V^{ij},V^{ik},V^{jk}$ are substituted from (\ref{vij}), hold identically
with respect to $V,V^i,V^j,V^k$.
\hfill\qedsymbol
\end{definition}

Less formally, the identities (\ref{consij}), (\ref{consijk}) can be
represented as
\begin{gather*}
 v^{ij}_1=T(r^{(ij)}_{[0]})=r^{(ij)}_{[1]},\quad\dots,\quad
 v^{ij}_{m-1}=T(r^{(ij)}_{[m-2]})=r^{(ij)}_{[m-1]},\\
 v^{ij}_m=T(r^{(ij)}_{[m-1]})=T_j(q^{(i)})=T_i(q^{(j)}),\\
 V^{ijk}=T_i(R^{(jk)})=T_j(R^{(ik)})=T_k(R^{(ij)})
\end{gather*}
where it is assumed that the shift operators $T,T_i,T_j$ act in virtue of
equations (\ref{vim}), (\ref{vij}).

One can prove that these conditions are not only necessary, but also
sufficient, on a lattice of arbitrary dimension, for the existence of common
solution of equations (\ref{vim}), (\ref{vij}) satisfying the generic initial
data defined on all 2-dimensional coordinate sublattices
$(\dots,0,n_i,0,\dots,0,n)$.

Two examples of multidimensionally consistent systems are presented in the rest
sections.

\section{B\"acklund transformation for the Bogoyavlensky lattice}\label{s:B}

There are several discrete equations of type (\ref{Qm}) related to the
Bogoyavlensky lattice
\begin{equation}\label{B.ut}
 \partial_tu=u(u_m+\dots+u_1-u_{-1}-\dots-u_{-m})
\end{equation}
and its modified versions. We restrict ourselves by consideration of just one
such equation which defines the B\"acklund transformation for the lattice
(\ref{B.ut}) in a potential form. The goal of this section is to prove the
property of multidimensional consistency, as formulated in the following
theorem, and to demonstrate that it is equivalent to the permutability property
of the B\"acklund transformations.

\begin{theorem}\label{th:B}
The system consisting of $m$-quad-equations
\begin{equation}\label{B.vi}
 (v_m-v^i_m)\cdots(v-v^i)=\beta^{(i)}v_{m-1}\cdots vv^i
\end{equation}
and of $m$-component quad-equations
\begin{gather}
\label{B.vij1}
 v^{ij}=\frac{\beta^{(j)}P^{(i)}v^j-\beta^{(i)}P^{(j)}v^i}
    {\beta^{(j)}P^{(i)}-\beta^{(i)}P^{(j)}},\quad
    P^{(i)}:=(v_{m-1}-v^i_{m-1})\dots(v-v^i),\\
\label{B.vijs}
  \frac{v^{ij}_s-v^j_s}{v^j_{s-1}}
 +\frac{v^j_s-v^i_s}{v_{s-1}}
 +\frac{v^i_s-v^{ij}_s}{v^i_{s-1}}=0,\quad s=1,\dots,m-1
\end{gather}
is multidimensionally consistent in the sense of Definition \ref{def:cons}.
Also, it is consistent with the lattice equation
\begin{equation}\label{B.vt}
 \partial_tv=v\left(\frac{v_m}{v}+\frac{v_{m-1}}{v_{-1}}
  +\cdots+\frac{v}{v_{-m}}\right)
\end{equation}
related with (\ref{B.ut}) by the substitution $u=v_m/v$.
\end{theorem}

In particular, at $m=1$ (the Volterra lattice case), equation (\ref{B.vi})
takes the form
\[
 (v_1-v^i_1)(v-v^i)=\beta^{(i)}vv^i
\]
and system (\ref{B.vij1}), (\ref{B.vijs}) is reduced to the quad-equation
$Q^0_1$
\[
 \beta^{(i)}(v-v^j)(v^i-v^{ij})=\beta^{(j)}(v-v^i)(v^j-v^{ij}).
\]
In the general case, equation (\ref{B.vijs}) without the restrictions on the
values of $s$ is a well-known 3-dimensional integrable equation equivalent to
the Hirota equation. The multidimensional consistency property for equations of
this type was studied in \cite{ABS12}. Equation (\ref{B.vij1}) is a constraint
which defines a reduction of this 3-dimensional equation to the $m$-component
2-dimensional one and equations (\ref{B.vi}) play the role of additional
constraints. One proof of the theorem can be obtained by verifying the
compatibility of these constraints with the 3-dimensional equation.

We will take another way which is completely within the 2-dimensional theory,
starting from the Darboux transformation for the discrete linear problem
\begin{equation}\label{B.psi}
 \psi_{-1}-u\psi_m=\lambda\psi
\end{equation}
associated with the Bogoyavlensky lattice.

\begin{statement}\label{st:B.Darboux}
Let the function $\psi(n)$ be a solution of equation (\ref{B.psi}) where
\begin{equation}\label{B.uf}
 u=f_m\cdots f_1(f-\alpha)
\end{equation}
and $u(n)\ne0$ for all $n$, then the function
\begin{equation}\label{B.Darboux}
 \tilde\psi=\frac{1}{f-\alpha}\Bigl(\frac{\alpha}{\lambda}\psi_{-1}-f\psi\Bigr)
\end{equation}
is a solution of equation
$\tilde\psi_{-1}-\tilde{u}\tilde\psi_m=\lambda\tilde\psi$ where
\begin{equation}\label{B.u1f}
 \tilde u=(f_m-\alpha)f_{m-1}\cdots f.
\end{equation}
\end{statement}
\begin{proof}
It is easy to prove that if $u,\tilde u$ are of the prescribed form then the
operators
\begin{gather*}
 L=T^{-1}-uT^m,\quad \tilde L=T^{-1}-\tilde uT^m,\quad
 A=\frac{1}{f-\alpha}\Bigl(\frac{\alpha}{\lambda}T^{-1}-f\Bigr),\\
 B=\frac{\lambda}{\alpha}
  +\frac{\alpha^{m-1}}{\lambda^{m-1}}(f-\alpha)T(1-Y^m)(1-Y)^{-1}
\end{gather*}
where
\[
 Y=\frac{\lambda}{\alpha}fT,\quad
 \mu=\frac{\alpha^m}{\lambda^m}-\frac{\lambda}{\alpha}
\]
satisfy the identities
\begin{gather*}
 B\cdot(f-\alpha)A=L-\lambda+\mu(f-\alpha),\\
 (f-\alpha)A\frac{1}{f-\alpha}B\cdot(f-\alpha)
 =\tilde L-\lambda+\mu(f-\alpha).
\end{gather*}
Here, we assume that $\alpha\ne0$ since otherwise the statement becomes
trivial. In order to make the computation, it is convenient to represent the
operator $A$ as $A=\dfrac{\alpha}{\lambda(f-\alpha)}(1-Y)T^{-1}$.

The first identity implies that any solution of equation $L\psi=\lambda\psi$
satisfies as well the equation
\[
 B\cdot(f-\alpha)A\psi=\mu(f-\alpha)\psi.
\]
Hence, $\tilde\psi=A\psi$ satisfies the equation
\[
 A\frac{1}{f-\alpha}B\cdot(f-\alpha)\tilde\psi=\mu\tilde\psi
\]
and the second identity implies $\tilde L\psi=\lambda\tilde\psi$.
\end{proof}

\begin{remark}
The condition $u(n)\ne0$ means, apparently, that $f(n)\ne0$, $f(n)\ne\alpha$,
for all $n$. In principle, this restriction may be waived, but it does not lead
to an essential generalization, rather it requires a lot of stipulations. So,
we will assume that this condition is fulfilled in what follows.
\end{remark}

Relations (\ref{B.uf}), (\ref{B.u1f}) play the role of the Miura type
transformations for the Bogoyavlensky lattice. Given the function $u$, any
function $f$ satisfying the equation $u=f_m\cdots f_1(f-\alpha)$ can be
represented as $f=\phi_{-1}/\phi$ where $\phi$ is a particular solution of
equation (\ref{B.psi}) at $\lambda=\alpha$. The lattice (\ref{B.ut}) is
equivalent to the compatibility condition of equation (\ref{B.psi}) and
\[
 \partial_t\psi=\psi_{-m-1}-(u+\cdots+u_{-m})\psi
\]
and this allows to obtain easily the modified lattice equation for the variable
$f$:
\begin{equation}\label{B.ft}
 \partial_tf=f(f-\alpha)(f_m\cdots f_1-f_{-1}\cdots f_{-m}).
\end{equation}
Vice versa, a direct computation proves that each of the substitutions
(\ref{B.uf}), (\ref{B.u1f}) maps solutions of equation (\ref{B.ft}) into
solutions of equation (\ref{B.ut}). The elimination of the variable $u$ brings
to the following statement.

\begin{statement}\label{st:B.ft}
Equation
\begin{equation}\label{B.f}
 (f_m-\alpha)f_{m-1}\cdots f=\tilde f_m\cdots\tilde f_1(\tilde f-\tilde\alpha)
\end{equation}
defines the $n$-part of the B\"acklund transformation for equation
(\ref{B.ft}), that is, differentiating of (\ref{B.f}) in virtue of (\ref{B.ft})
and an analogous lattice equation for $\tilde f,\tilde\alpha$ yields a relation
which holds identically in virtue of equation (\ref{B.f}) itself.
\end{statement}

In other words, the lattice equation (\ref{B.ft}) defines a higher symmetry for
(\ref{B.f}). The existence of just one such symmetry is a nontrivial fact which
allows to say about the integrability of the equation. But, what is about the
3D-consistency? As we have already mentioned in the introduction, although
(\ref{B.f}) belongs to the class of equations (\ref{Qm}) under consideration,
but the variables $f$ are naturally associated with the edges of the lattice
rather than the vertices, and the Definition from the previous section is not
suitable. In this situation, the 3D-consistency should be understood in the
sense of the Yang--Baxter maps. We do not give one more general definition,
since the variables $f$ play an auxiliary role in our consideration; the
required property is formulated in Statement \ref{st:B.NSP} below. In order to
prove it, we use the matrix representation of the linear problem (\ref{B.psi})
and its Darboux transformation (\ref{B.Darboux})
\begin{equation}\label{B.Psi}
 \Psi_{-1}=U\Psi,\quad \widetilde\Psi=F\Psi \quad\Rightarrow\quad
 \widetilde UF=F_{-1}U
\end{equation}
where $\Psi$ is a $(m+1)$-dimensional vector and $U,F$ are $(m+1)\times(m+1)$
matrices:
\begin{equation}\label{B.matU}
 \Psi=\begin{pmatrix}\psi_m\\ \vdots\\ \psi\end{pmatrix},\quad
 U=\begin{pmatrix}
  0 & 1 & \dots  & 0\\
    &   & \ddots &  \\
  0 & 0 & \dots  & 1\\
  u & 0 & \dots  & \lambda
  \end{pmatrix},
\end{equation}
\begin{equation}\label{B.matF}
 F=\begin{pmatrix}
  -\frac{f_m}{f_m-\alpha} &
   \frac{\alpha}{\lambda(f_m-\alpha)} & 0 &\dots & 0 \\
  0 & -\frac{f_{m-1}}{f_{m-1}-\alpha} &
       \frac{\alpha}{\lambda(f_{m-1}-\alpha)} &\dots & 0 \\
    &   &\ddots &\ddots  & \\
  0 & 0 &\dots  & -\frac{f_1}{f_1-\alpha} &
                   \frac{\alpha}{\lambda(f_1-\alpha)}\\
  \frac{\alpha}{\lambda}f_m\cdots f_1 & 0 &\dots  & 0 & -1
  \end{pmatrix}.
\end{equation}
These matrices are easily derived from equations (\ref{B.psi}),
(\ref{B.Darboux}). Statement \ref{st:B.Darboux} means exactly that relations
(\ref{B.uf})--(\ref{B.u1f}) are equivalent to the matrix equations
(\ref{B.Psi}). The permutability property of two Darboux transformations of
this form
\[
 \widetilde\Psi=F[f,\alpha,\lambda]\Psi,\quad
 \widehat\Psi=F[g,\beta,\lambda]\Psi
\]
is expressed by equation
\begin{equation}\label{B.FF}
  F[\tilde g,\beta,\lambda]\,F[f,\alpha,\lambda]
 =F[\hat f,\alpha,\lambda]\,F[g,\beta,\lambda].
\end{equation}
Due to the special structure of the matrices this equation can be easily solved
and a direct computation brings to the following statement. The proof scheme of
the identity (\ref{B.TTF}) based on the refactorization of the product of
matrix triple is well known (see e.g. \cite{AY}).

\begin{statement}\label{st:B.NSP}
Let matrices $F$ be of the form (\ref{B.matF}) and $\alpha\ne\beta$ then
equation (\ref{B.FF}), viewed as an identity with respect to $\lambda$, is
uniquely solved with respect to $\hat f,\tilde g$ for the generic values of
$f,g$. Denote $f=f^{(i)}$, $\alpha=\alpha^{(i)}$, $g=f^{(j)}$,
$\beta=\alpha^{(i)}$, then the solution is written as the $m$-component
mapping
\begin{equation}\label{B.Tf}
\begin{aligned}
 & T_j(f^{(i)}_1)=\frac{(\alpha^{(j)}f^{(j)}_m\cdots f^{(j)}_1
     -\alpha^{(i)}f^{(i)}_m\cdots f^{(i)}_1)(f^{(j)}_1-\alpha^{(j)})}
    {f^{(i)}_m\cdots f^{(i)}_2(\alpha^{(j)}f^{(i)}_1-\alpha^{(i)}f^{(j)}_1)},\\
 & T_j(f^{(i)}_s)=\frac{f^{(i)}_s(f^{(j)}_s-\alpha^{(j)})
    (\alpha^{(j)}f^{(i)}_{s-1}-\alpha^{(i)}f^{(j)}_{s-1})}
   {(f^{(j)}_{s-1}-\alpha^{(j)})(\alpha^{(j)}f^{(i)}_s-\alpha^{(i)}f^{(j)}_s)},
  \quad s=2,\dots,m.
\end{aligned}
\end{equation}
This mapping satisfies the 3D-consistency property
\begin{equation}\label{B.TTF}
 T_kT_j(f^{(i)}_s)=T_jT_k(f^{(i)}_s),\quad s=1,\dots,m.
\end{equation}
\end{statement}
\begin{proof}
One can prove that equation (\ref{B.FF}) is uniquely solvable also with respect
to the variables $\hat f,g$, for given $\tilde g,f$. Moreover, the following
property is fulfilled: if parameters $\alpha,\beta,\gamma$ are distinct and
\begin{equation}\label{B.P}
 F[h,\gamma,\lambda]F[g,\beta,\lambda]F[f,\alpha,\lambda]
 =F[h',\gamma,\lambda]F[g',\beta,\lambda]F[f',\alpha,\lambda]
\end{equation}
then $h'=h$, $g'=g$, $f'=f$. This follows from the fact that the matrix
(\ref{B.matF}) is uniquely defined by its one-dimensional kernel at the
spectral parameter value $\lambda=\alpha$:
\[
 \ker F[f,\alpha,\alpha]=(1,f_m,f_mf_{m-1},\dots,f_m\cdots f_1)^\top.
\]
Now, let us assign the parameters $\alpha,\beta,\gamma$ to three coordinate
directions and consider refactorizations of the matrix
$F[h,\gamma,\lambda]F[g,\beta,\lambda]F[f,\alpha,\lambda]$ corresponding to all
permutations of these parameters. Due to the property (\ref{B.P}), each
permutation corresponds to a unique set of values $f,g,h$ which are naturally
associated with edges of a cube. Moreover, the value at each edge is obtained
by two different sequences of pairwise refactorizations which is equivalent to
the property (\ref{B.TTF}).
\end{proof}

\begin{proof}[Proof of the Theorem \ref{th:B}]
The above statements imply the consistency, in the lattice of arbitrary
dimension, of equations
\begin{equation}\label{B.uf'}
 u=f^{(i)}_m\cdots f^{(i)}_1(f^{(i)}-\alpha^{(i)}),\quad
 u^i=(f^{(i)}_m-\alpha^{(i)})f^{(i)}_{m-1}\cdots f^{(i)},
\end{equation}
the mapping (\ref{B.Tf}) and the differential-difference equations
(\ref{B.ut}), (\ref{B.ft}). The final step is to pass from the edge-type
variables $f^{(i)}$ to the variables $v$ associated with the vertices of the
lattice. To do this, notice that (\ref{B.uf'}) implies the multiplicative
conservation law
\begin{equation}\label{B.uuff}
 \frac{u^i}{u}=\frac{(f^{(i)}_m-\alpha^{(i)})/f^{(i)}_m}{(f-\alpha)/f}
\end{equation}
and this allows to introduce (uniquely up to a constant factor) the potential
variable $v$, according to the equations
\[
 u=\frac{v_m}{v},\quad f^{(i)}=\frac{\alpha^{(i)}v}{v-v^i}.
\]
Under this change, equations (\ref{B.ut}), (\ref{B.ft}) turn into equation
(\ref{B.vt}), relations (\ref{B.uf'}) turn into the $m$-quad-equation
(\ref{B.vi}) and the mapping (\ref{B.Tf}) turn into the $m$-component
quad-equation (\ref{B.vij1}), (\ref{B.vijs}). The parameter $\alpha$ enters
into new equations only in power of $m+1$, so we denote additionally
$(\alpha^{(i)})^{m+1}=\beta^{(i)}$.
\end{proof}

\begin{remark}
The potential can be introduced in different ways, for instance, instead of
(\ref{B.uuff}) one can consider the conservation law
\[
 \frac{u^i_1}{u}=\frac{f^{(i)}_{m+1}-\alpha^{(i)}}{f-\alpha}
\]
which leads to the substitution
\[
 u=\frac{w_{m+1}}{w},\quad f^{(i)}=\frac{w^i_1}{w}+\alpha^{(i)}.
\]
The variable $w$ satisfies the equation
\[
 w_m\cdots w=(w^i_m+\alpha^{(i)}w_{m-1})\cdots(w^i_1+\alpha^{(i)}w)w^i
\]
with the properties similar to the properties of equation (\ref{B.vi}).
\end{remark}

\begin{remark}
From the point of view of logical simplicity, introducing of the potential
seems to be a redundant step, because we already have the variable defined in
the vertices of the lattice---this is $u$, the coefficient of the original
linear problem (\ref{B.psi}). This variable satisfies some multiquad-equation,
indeed, but the problem is that it is too complicated and it is hardly possible
to write it down in general form for all $m$. In contrast to equations for $f$,
$v$ or $w$, this equation is not affine-linear with respect to the involved
values. For instance, in the simplest case $m=1$,  the relations (\ref{B.uf})
take the form
\[
 u=f_1(f-\alpha),\quad \tilde u=(f_1-\alpha)f;
\]
from here one finds
\[
 f=\frac{1}{2\alpha}(u-\tilde u+\alpha^2
  +\sqrt{(\tilde u-u)^2+2\alpha^2(\tilde u+u)+\alpha^4}),\quad
 f_1=f+\frac{1}{\alpha}(\tilde u-u)
\]
and elimination of $f$ brings to the quad-equation
\begin{multline*}\quad
 \tilde u-u+\sqrt{(\tilde u-u)^2+2\alpha^2(\tilde u+u)+\alpha^4})\\
  = u_1-\tilde u_1
  +\sqrt{(\tilde u_1-u_1)^2+2\alpha^2(\tilde u_1+u_1)+\alpha^4}). \quad
\end{multline*}
It can be brought to a polynomial form quadratic with respect to each variable;
the general theory of such equations was developed in \cite{AN}. Analogously,
at $m=2$, a polynomial equation cubic with respect to each
variable appears after elimination of $f,\dots,f_4$ from equations
\begin{gather*}
 u_s=f_{s+2}f_{s+1}(f_s-\alpha),\quad \tilde u_s=(f_{s+2}-\alpha)f_{s+1}f_s,
 \quad s=0,1,2.
\end{gather*}
\end{remark}

\section{B\"acklund transformation for the discretization of Sawada--Kotera
equation}\label{s:dSK}

This section is devoted to the discrete equations related to the lattice
equation
\begin{equation}\label{dSK.ut}
 \partial_tu=u^2(u_m\cdots u_1-u_{-1}\cdots u_{-m})
  -u(u_{m-1}\cdots u_1-u_{-1}\cdots u_{1-m}).
\end{equation}
The main result is the following Theorem, which presents the B\"acklund
transformation and nonlinear superposition formula for the potential form of
(\ref{dSK.ut}).

\begin{theorem}\label{th:dSK}
The following system of equations is multidimensionally consistent:
\begin{gather}
\label{dSK.vi}
 (v^i_m-v)v^i_{m-1}\cdots v^i=\alpha^{(i)}(v_m-v^i)v_{m-1}\cdots v,\\
\label{dSK.vij}
 v^{ij}_s=v_s\frac{P^{(i,j)}_{s+1}}{P^{(i,j)}_s},\quad s=0,\dots,m-1
\end{gather}
where
\begin{gather*}
\begin{aligned}
 P^{(i,j)}_s &= (\alpha^{(i)}v^j_{m-1}\cdots v^j
    -\alpha^{(j)}v^i_{m-1}\cdots v^i)v_{m-1}\cdots v\\
 &\quad
  +\sum^{m-1}_{r=0}\gamma^{(i,j)}_{r,s}(v_{m-1}\cdots v_{r+1})^2
   v_r(v^i_r-v^j_r)v^i_{r-1}\cdots v^iv^j_{r-1}\cdots v^j,
\end{aligned}\\
 \gamma^{(i,j)}_{r,s}=
  \begin{cases}
   1 & s\le r,\\
   \alpha^{(i)}\alpha^{(j)} & r<s.
  \end{cases}
\end{gather*}
It is also consistent with the lattice equation
\begin{equation}\label{dSK.vt}
 \partial_tv=\frac{v_{m-1}\cdots v}{v_{-1}\cdots v_{-m}}(v_m-v_{-m})
\end{equation}
related with (\ref{dSK.ut}) by the substitution $u=v_m/v$.
\end{theorem}

Notice, that equation (\ref{dSK.ut}) at $m=1$ coincides with the modified
Volterra lattice $\partial_tu=u^2(u_1-u_{-1})$ and system (\ref{dSK.vij})
turns into equation
\[
 (\alpha^{(i)}-1)(v^{ij}v^j-\alpha^{(j)}v^iv)=
 (\alpha^{(j)}-1)(v^{ij}v^i-\alpha^{(i)}v^jv)
\]
which is equivalent to $H^0_3$ quad-equation. Equation (\ref{dSK.ut}) at $m>1$
was studied in \cite{AP}, where it was shown that its continuous limit
coincides with the Sawada--Kotera equation. This equation contains, as
particular cases, the modified Bogoyavlensky lattices
\[
 \partial_{t'}u=u^2(u_m\cdots u_1-u_{-1}\cdots u_{-m}),\quad
 \partial_{t''}u=u(u_{m-1}\cdots u_1-u_{-1}\cdots u_{1-m})
\]
which can be obtained by scaling $u,t$ and passing to the limit. However, these
flows do not commute and the integrability of their linear combination is a
nontrivial fact which follows from the existence of the Lax pair \cite{AP}
\begin{gather}
\label{dSK.psi}
 u\psi_{m+1}-\psi_m+\lambda(\psi_1-u\psi)=0,\\
\label{dSK.psit}
 \partial_t\psi=u_{-1}\cdots u_{-m}(\lambda\psi_{-m}-\lambda^{-1}\psi_m).
\end{gather}
It is easy to verify that equation (\ref{dSK.ut}) is equivalent to the
compatibility condition for these linear equations.

Derivation of the Darboux transformation for equation (\ref{dSK.psi}) is the
main step in the proof of Theorem \ref{th:dSK}. The action of this
transformation on function $\psi$ is defined in a more complicated way than in
the case of the Bogoyavlensky lattice, where operators of the first order were
used (see Statement \ref{st:B.Darboux}). Nevertheless, on the level of
coefficients of the linear problem the Darboux transformation is described, as
before, by the pair of Miura type substitutions,
\begin{equation}\label{dSK.uf}
 u=\frac{g}{G}f,\quad \tilde u=\frac{g}{G}f_m
\end{equation}
where
\[
 g=f_{m-1}\cdots f_1-\alpha,\quad G=f_m\cdots f-\alpha
\]
(this notation will be used throughout the section).

\begin{statement}\label{st:dSK.Darboux}
Let functions $u,\tilde u$ be of the form (\ref{dSK.uf}), $\psi$ be a solution
of equation (\ref{dSK.psi}) and $\tilde\psi$ be a solution of similar
equation
\begin{equation}\label{dSK.tpsi}
 \tilde u\tilde\psi_{m+1}-\tilde\psi_m
 +\lambda(\tilde\psi_1-\tilde u\tilde\psi)=0.
\end{equation}
Then both formulae $\chi=\psi_1-f\psi$ and $\chi=\tilde\psi-f\tilde\psi_1$
define a solution of equation
\begin{equation}\label{dSK.chi}
 (1-f_mf)(f_1g_1\chi_{m+1}+\lambda G_1\chi_1)=
 (1-f_{m+1}f_1)(G\chi_m+\lambda f_mg\chi).
\end{equation}
\end{statement}
\begin{proof}
Let us consider the first substitution:
\begin{align*}
 &\chi=\psi_1-f\psi,\\
 &\chi_1=\psi_2-f_1\psi_1,\\
 &\chi_m=\psi_{m+1}-f_m\psi_m
   =\lambda\psi-\frac{\lambda}{u}\psi_1+\Bigl(\frac{1}{u}-f_m\Bigr)\psi_m,\\
 &\chi_{m+1}=\lambda\psi_1-\frac{\lambda}{u_1}\psi_2
  +\Bigl(\frac{1}{u_1}-f_{m+1}\Bigr)(\chi_m+f_m\psi_m).
\end{align*}
Equation for $\chi$ is obtained by solving first three equations with respect
to $\psi_1,\psi_2,\psi_m$ and plugging the result into the last one. A
straightforward (but rather lengthy) computation proves that the terms with
$\psi$ cancel out (taking the relation $u=gf/G$ into account) and equation
(\ref{dSK.chi}) appears.

The second substitution $\chi=\tilde\psi-f\tilde\psi_1$ can be checked
analogously, but it is better to make use of the reflection
\[
 f(n)\to f(-n),\quad \psi(n)\leftrightarrow\tilde\psi(1-n),\quad
 \chi(n)\to\chi(-n),\quad \lambda\to\frac{1}{\lambda}.
\]
Under this transformation, equations (\ref{dSK.psi}) and (\ref{dSK.tpsi}) turn
into each other, as well as the substitutions $\chi=\psi_1-f\psi$ and
$\chi=\tilde\psi-f\tilde\psi_1$, and equation (\ref{dSK.chi}) does not change.
Indeed, equation (\ref{dSK.psi})
\[
 \frac{f_{m-1}\cdots f_1-\alpha}{f_m\cdots f-\alpha}f
 (\psi_{m+1}-\lambda\psi)=\psi_m-\lambda\psi_1
\]
turns into
\[
 \frac{f_{-m+1}\cdots f_{-1}-\alpha}{f_{-m}\cdots f-\alpha}f
 (\tilde\psi_{-m}-\frac{1}{\lambda}\tilde\psi_1)
 =\tilde\psi_{1-m}-\frac{1}{\lambda}\tilde\psi
\]
and applying of $-\lambda T^m$ results in equation (\ref{dSK.tpsi})
\[
 \frac{f_{m-1}\cdots f_1-\alpha}{f_m\cdots f-\alpha}f_m
 (\tilde\psi_{m+1}-\lambda\tilde\psi)
 =\tilde\psi_m-\lambda\tilde\psi_1.
\]
To verify the invariance of equation (\ref{dSK.chi}), it is sufficient to
notice that the reflection sends $g_k$ to $g_{-m-k}$ and $G_k$ to $G_{-m-k}$.
\end{proof}

Equation (\ref{dSK.chi}) looks awkward, but we need it only as an intermediate
between equations (\ref{dSK.psi}) and (\ref{dSK.tpsi}). It is important only
that {\em both} these equations admit first order Darboux transformations into
one and the same equation, so it is possible to define a Darboux transformation
between them:
\[
 \chi=(T-f)\psi,\quad \chi=(1-fT)\tilde\psi
 \quad\Rightarrow\quad \tilde\psi=(1-fT)^{-1}(T-f)\psi.
\]
The operator $1-fT$ is invertible in virtue of equation (\ref{dSK.psi}) and the
operator $(1-fT)^{-1}(T-f)$ is equivalent to some operator of $m$-th order (its
explicit form is rather bulky, but we do not need it). In the matrix
representation, the Darboux transformation is defined by equations
\[
 \Psi_1=U\Psi,\quad \tilde\Psi=F\Psi,\quad F=B^{-1}A
\]
where
\begin{gather}
\nonumber
 \Psi=\begin{pmatrix}\psi\\ \vdots\\ \psi_m\end{pmatrix},\quad
 U=\begin{pmatrix}
  0 & 1 & \dots  & 0 & 0\\
  0 & 0 & \ddots & 0 & 0\\
    &   & \ddots & \ddots & \\
  0 & 0 & \dots  & 0 & 1\\
  \lambda & -\dfrac{\lambda}{u} & \dots & 0 & \dfrac{1}{u}
  \end{pmatrix},\\
\label{dSK.A}
 A=\begin{pmatrix}
  -f &    1 & \dots  & 0 & 0\\
   0 & -f_1 & \ddots & 0 & 0\\
     &      & \ddots & \ddots   &  \\
   0 &    0 & \dots  & -f_{m-1} & 1\\
  \lambda & -\dfrac{\lambda G}{gf} & \dots & 0 & \dfrac{G}{gf}-f_m
  \end{pmatrix},\\
\label{dSK.B}
 B=\begin{pmatrix}
   1 & -f & \dots  & 0 & 0\\
   0 &  1 & \ddots & 0 & 0\\
     &    & \ddots & \ddots & \\
   0 &  0 & \dots  & 1 & -f_{m-1} \\
  -\lambda f_m & \lambda\dfrac{G}{g} & \dots & 0 & 1-\dfrac{G}{g}
 \end{pmatrix}.
\end{gather}
The consistency condition $\widetilde UF=F_1U$ is exactly equivalent to the
relations (\ref{dSK.uf}). Pay attention that although the variable $f_m$ enters
the matrices $A$ and $B$, the matrix $F$ does not depend on it.

A more complicated structure of the Darboux matrix $F$ is the only essential
difference from the example considered in the previous section. The general
scheme of the proof remains the same. Any function $f$ satisfying equation
$u=gf/G$ can be represented as $f=\phi_1/\phi$ where $\phi$ is a particular
solution of equation (\ref{dSK.psi}) at $\lambda=\alpha$. The use of equation
(\ref{dSK.psit}) brings to the modified lattice equation
\begin{equation}\label{dSK.ft}
 \partial_tf=f\frac{g_{-1}\cdots g_{-m+1}}{G\cdots G_{-m}}
  \bigl(gg_{-m}(G-G_{-m})-GG_{-m}(g-g_{-m})\bigr).
\end{equation}
Vice versa, one can verify intermediately that each substitution (\ref{dSK.uf})
maps a solution of equation (\ref{dSK.ft}) into a solution of (\ref{dSK.ut}).
This proves the following statement.

\begin{statement}\label{st:dSK.ft}
Equation
\begin{equation}\label{dSK.f}
 \frac{(\tilde f_{m-1}\cdots\tilde f_1-\tilde\alpha)\tilde f}
   {\tilde f_m\cdots\tilde f-\tilde\alpha}
 =\frac{f_m(f_{m-1}\cdots f_1-\alpha)}{f_m\cdots f-\alpha}
\end{equation}
defines the $n$-part of the B\"acklund transformation for equation
(\ref{dSK.ft}), that is, the relation obtained by differentiating of this
equation in virtue of (\ref{dSK.ft}) and a similar lattice equation for $\tilde
f,\tilde\alpha$ is fulfilled identically in virtue of the equation itself.
\end{statement}

\begin{remark}
An interesting property of equation (\ref{dSK.f}) at $\tilde\alpha=\alpha$ is
that it admits a reduction to an equation of order $m-1$. More precisely, one
can prove that it can be brought to the form $\tilde fT(Q)=f_mQ$ where $Q$ is a
polynomial depending on $f,\dots,f_{m-1},\tilde f,\dots,\tilde f_{m-1}$, so
that equation $Q=0$ defines a special solution of (\ref{dSK.f}). Moreover,
equation $Q=0$ is also consistent with the lattice equation (\ref{dSK.ft}). For
instance, in the simplest case $m=2$, the quad-equation appears (a discrete
analog of the Tzitzeica equation)
\[
 \tilde ff_1(\alpha^{-1}\tilde f_1f-\tilde f_1-f)+f_1+\tilde f-\alpha=0
\]
which is consistent with the flow
\[
 \partial_tf=\frac{f(f-\alpha)}{f_1ff_{-1}-\alpha}\left(
 \frac{f(f_1-\alpha)(f_{-1}-\alpha)(f_2f_1-f_{-1}f_{-2})}
      {(f_2f_1f-\alpha)(ff_{-1}f_{-2}-\alpha)}-f_1+f_{-1}\right).
\]
Equation (\ref{B.f}) admits an analogous lowering of the order.
\end{remark}

Refactorization of the above Darboux matrices is a rather difficult task.
Nevertheless, it is possible to write the general answer for arbitrary $m$ and
we arrive to the following statement. The proof of 3D-consistency follows from
the same general arguments as in the case of Statement \ref{st:B.NSP}.

\begin{statement}\label{st:dSK.NSP}
Let $F[f,\alpha,\lambda]=B^{-1}A$ where $A,B$ are matrices of the form
(\ref{dSK.A}), (\ref{dSK.B}) and $\alpha^{(i)}\ne\alpha^{(j)}$. Then
equation
\[
  F[T_i(f^{(j)}),\alpha^{(j)},\lambda]\,F[f^{(i)},\alpha^{(i)},\lambda]
 =F[T_j(f^{(i)}),\alpha^{(i)},\lambda]\,F[f^{(j)},\alpha^{(j)},\lambda]
\]
considered as an identity with respect to $\lambda$ is uniquely solved with
respect to $T_i(f^{(j)}),T_j(f^{(i)})$, for generic values of $f^{(i)}$,
$f^{(j)}$. This gives rise to the $m$-component mapping
\begin{equation}\label{dSK.Tf}
 T_j(f^{(i)}_s)=\frac{P^{(i,j)}_{s+1}}{f^{(j)}_sP^{(i,j)}_s},
 \quad s=0,\dots,m-1
\end{equation}
where
\begin{gather*}
\begin{aligned}
 P^{(i,j)}_s &= \alpha^{(i)}f^{(j)}_{m-1}\cdots f^{(j)}
    -\alpha^{(j)}f^{(i)}_{m-1}\cdots f^{(i)}\\
 &\qquad
  +\sum^{m-1}_{r=0}\gamma^{(i,j)}_{r,s}(f^{(i)}_r-f^{(j)}_r)
   f^{(i)}_{r-1}\cdots f^{(i)}f^{(j)}_{r-1}\cdots f^{(j)},
\end{aligned}\\
 \gamma^{(i,j)}_{r,s}=
  \begin{cases}
   1 & s\le r,\\
   \alpha^{(i)}\alpha^{(j)} & r<s.
  \end{cases}
\end{gather*}
This mapping satisfies the 3D-consistency property
\[
 T_kT_j(f^{(i)}_s)=T_jT_k(f^{(i)}_s),\quad s=0,\dots,m-1.
\]
\end{statement}

Now, the proof of Theorem \ref{th:dSK} is obtained by introducing of the
potential $v$ according to equations
\[
 u=\frac{v_m}{v},\quad f^{(i)}=\frac{v^i}{v}.
\]
Under these substitutions, equations (\ref{dSK.uf}) turn into $m$-quad-equation
(\ref{dSK.vi}), the mapping (\ref{dSK.Tf}) turn into $m$-component
quad-equation (\ref{dSK.vij}), and the lattice equations (\ref{dSK.ut}),
(\ref{dSK.ft}) turn into equation (\ref{dSK.vt}). The consistency follows from
the proven statements.

\section*{Acknowledgements}
\addcontentsline{toc}{section}{Acknowledgements}

Research for this article was supported by grants RFBR 13-01-00402a and
NSh--5377.2012.2.

\phantomsection

\end{document}